\documentclass[oneside, reqno]{amsart}

\usepackage{amsmath, amssymb, amsthm}
\usepackage{enumitem}

\usepackage[dvipsnames]{xcolor}
\usepackage{lmodern}
\usepackage[T1]{fontenc}
\usepackage{microtype}

\newcommand{\alias}[1]{\providecommand{#1}{}\renewcommand{#1}}

\DeclareMathOperator\Tr{Tr}

\newcommand{\tot}{\tfrac{1}{2}} 
\newcommand{\oo}[1]{\tfrac{1}{#1}}

\newcommand{\abs}[1]{\left| #1 \right|} 
\newcommand{\babs}[1]{\big| #1 \big|} 

\newcommand{\set}[1]{\{#1\}} 

\newcommand{\norm}[1]{{||#1||}} 
\newcommand{\Bnorm}[1]{{\Big|\Big|#1\Big|\Big|}} 

\newcommand{\ft}[2]{#1\dots#2}
\renewcommand{\ft}[2]{#1,\dots,#2}

\newcommand{\prf}[1]{ \{ #1 \}_{t\in [0,T]}}

\providecommand{\R}{} \renewcommand{\R}{{\mathbb R}}

\providecommand{\N}{} \renewcommand{\N}{{\mathbb N}}

\newcommand{\PP}{{\mathbb P}}

\newcommand{\EE}{{\mathbb E}}

\newcommand{\FFF}{{\mathbb F}}

\newcommand{\bmu}{\boldsymbol{\mu}}
\newcommand{\ba}{{\mathrm{ba}}}

\newcommand{\eps}{\varepsilon}
\newcommand{\ld}{\lambda}
\newcommand{\Ld}{\Lambda}
\newcommand{\gm}{\gamma}
\newcommand{\vp}{\varphi}

\newcommand{\el}{{\mathbb L}} 

\newcommand{\lone}{\el^1}

\newcommand{\linf}{\el^{\infty}}

\newcommand{\define}[1]{{\textbf{#1}}}

\newcounter{notecounter}

\newcommand{\efor}{\text{ for }}
\newcommand{\eforall}{\text{ for all }}

\newcommand{\eand}{\text{ and }}

\newcommand{\itos}{It\^ o's}

\newcommand\bb{{\mathbb b}}

\newcommand\tc{{\tilde{c}}}
\newcommand\hc{{\hat{c}}}

\newcommand\te{{\tilde{e}}}

\newcommand\tf{{\tilde{f}}}

\newcommand\bg{{\mathbb g}}

\newcommand\bi{{\mathbb i}}

\newcommand\tx{{\tilde{x}}}

\newcommand\sA{{\mathcal A}}

\newcommand\bA{{\mathbb A}}

\newcommand\sC{{\mathcal C}}

\newcommand\hD{{\hat{D}}}

\newcommand\sF{{\mathcal F}}

\newcommand\sH{{\mathcal H}}

\newcommand\sK{{\mathcal K}}

\newcommand\sM{{\mathcal M}}

\newcommand\bM{{\mathbb M}}

\newcommand\sP{{\mathcal P}}

\newcommand\bR{{\mathbb R}}

\newcommand\sS{{\mathcal S}}

\newcommand\bS{{\mathbb S}}
\newcommand\tS{{\tilde{S}}}

\newcommand\tV{{\tilde{V}}}

\newcommand\hX{{\hat{X}}}

\numberwithin{equation}{section}
\theoremstyle{plain}                
\newtheorem{theorem}{Theorem}[section]

\newtheorem{proposition}[theorem]{Proposition}
\newtheorem{corollary}[theorem]{Corollary}

\theoremstyle{definition}           
\newtheorem{definition}[theorem]{Definition}

\newtheorem{assumption}[theorem]{Assumption}

\theoremstyle{remark}
\newtheorem{remark}[theorem]{Remark}

\let\oldsum\sum
\renewcommand{\sum}{\textstyle\oldsum}
\let\oldprod\prod
\renewcommand{\prod}{\textstyle\oldprod}

\alias{\Bro}{B^{\rho}}
\alias{\Dl}{\varDelta}
\alias{\Gm}{\varGamma}
\alias{\Mez}{M^{Z,\eta}}
\alias{\Mez}{M^{Z,\eta}}
\alias{\Me}{M^{\eta}}
\alias{\Me}{M^{\eta}}
\alias{\Mz}{M^{Z}}
\alias{\Sin}{\sS^{\infty}}
\alias{\So}{S^1}
\alias{\Sz}{S^0}
\alias{\Wi}{W^{1,\infty}}
\alias{\Wro}{W^{\rho}}
\alias{\Xdc}{X^{\Delta,c}}
\alias{\YN}{\upN{Y}}
\alias{\YiN}{Y^{i,(N)}}
\alias{\ZN}{\upN{Z}}
\alias{\ZiN}{Z^{i,(N)}}
\alias{\aN}{\upN{a}}
\alias{\aeta}{\abs{\eta}}
\alias{\ai}{\alpha^i}
\alias{\al}{\alpha}
\alias{\asig}{\abs{\sigma}}
\alias{\bA}{\overline{A}}
\alias{\bMez}{\bar{M}^{Z,\eta}}
\alias{\bMez}{\bar{M}^{Z,\eta}}
\alias{\bMe}{\bar{M}^{\eta}}
\alias{\bMe}{\bar{M}^{\eta}}
\alias{\bMz}{\bar{M}^{Z}}
\alias{\bR}{\bar{R}}
\alias{\bSo}{\bS^1}
\alias{\bSz}{\bS^0}
\alias{\bS}{\overline{S}}
\alias{\bara}{\bar{\alpha}}
\alias{\barb}{\bar{\beta}}
\alias{\ba}{\bar{\alpha}}
\alias{\bb}{\bar{\beta}}
\alias{\bg}{\bar{g}}
\alias{\bi}{\beta^i}
\alias{\bld}{\bar{\lambda}}
\alias{\bmo}{\ensuremath{\mathrm{bmo}}}
\alias{\bmu}{\bar{\mu}}
\alias{\bsalpha}{\boldsymbol{\alpha}}
\alias{\bsbeta}{\boldsymbol{\beta}}
\alias{\bsigma}{\bar{\sigma}}
\alias{\btS}{\tilde{\bar{S}}}
\alias{\cdoze}{\cdo{\zeta},\cdo{\eta}}
\alias{\cdoz}{\cdo{\zeta}}
\alias{\cmu}{\check{\mu}}
\alias{\dl}{\delta}
\alias{\eL}{\el^{\Psi}}
\alias{\eM}{\bM^{\Phi}}
\alias{\ehnu}{\hat{\nu}^{\eps}}
\alias{\eie}{\text{ i.e., }}
\alias{\epnu}{\hat{\nu}^{\eps, \perp}}
\alias{\fNp}{f^{(N')}}
\alias{\fr}{\tfrac}
\alias{\f}{\frac}
\alias{\hD}{\hat{D}}
\alias{\halpha}{\hat{\alpha}}
\alias{\hbsc}{\hat{\boldsymbol{c}}}
\alias{\hbspi}{\hat{\boldsymbol{\pi}}}
\alias{\hc}{\hat{c}}
\alias{\heta}{\hat{\eta}}
\alias{\heta}{\hat{\eta}}
\alias{\heta}{\hat{\eta}}
\alias{\hmu}{\hat{\mu}}
\alias{\hnu}{\hat{\nu}}
\alias{\hpi}{\hat{\pi}}
\alias{\hpi}{\hat{\pi}}
\alias{\hth}{\hat{\th}}
\alias{\lob}{\lone_{\beta}}
\alias{\mfz}{\mathfrak{Z}}
\alias{\peta}{\hat{\eta}^{\perp}}
\alias{\sAN}{\sA^{N}}
\alias{\sAs}{\sA^{\mathrm{smp}}}
\alias{\sCe}{\sC^{\mathrm{ess}}}
\alias{\sCs}{\sC^{\mathrm{smp}}}
\alias{\sKs}{\sK^{\mathrm{smp}}}
\alias{\saif}{\sA^{\infty}_f}
\alias{\scif}{\sC^{\infty}_f}
\alias{\sigN}{\upN{\sigma}}
\alias{\srho}{\sqrt{1-\rho^2}}
\alias{\tDelta}{\tilde{\Delta}}
\alias{\tS}{\tilde{S}}
\alias{\tbS}{\widetilde{\overline{S}}}
\alias{\tc}{\tilde{c}}
\alias{\te}{\tilde{e}}
\alias{\tf}{\tfrac}
\alias{\tho}{\theta^1}
\alias{\th}{\theta}
\alias{\tmu}{\tilde{\mu}}
\alias{\tpi}{\tilde{\pi}}
\alias{\tso}{t^1}
\alias{\tsz}{t^0}
\alias{\tvp}{\tilde{\vp}}
\alias{\uNP}{u^{(N+1)}}
\alias{\uN}{u^{(N)}}
\alias{\ueta}{\underline{\eta}}
\alias{\usM}{\underline{\sM}}
\alias{\uzeta}{\underline{\zeta}}
\alias{\vN}{\upN{v}}
\alias{\wN}{\upN{w}}
\alias{\Urep}{U_{\text{rep}}}
\alias{\rrep}{r_{\text{rep}}}
\alias{\lrep}{\lambda_{\text{rep}}}
\newcommand{\cdo}[1]{#1_{\cdot}}

\newcommand{\ioNz}[1]{\iota_{N_0}(#1)}
\newcommand{\ioN}[1]{\iota_N(#1)}
\newcommand{\upN}[1]{#1^{(N)}}
\alias{\osig}{\overline{\sigma}}
\alias{\usig}{\underline{\sigma}}
\alias{\oC}{\overline{\sigma}}
\alias{\uC}{\underline{\sigma}}
\alias{\tx}{t,x}
\alias{\txp}{t',x'}
\alias{\sxp}{s,x'}
\alias{\sumi}{\sum_{i=1}^I}
\alias{\suml}{\sum_{l=1}^I}
\alias{\mue}{\mu_e}
\alias{\mueN}{\mu_{e,N}}
\alias{\sigmae}{\sigma_e}
\alias{\Rd}{\R^{d}}
\alias{\Rid}{\R^{1\times d}}
\alias{\Rdi}{\R^{d\times 1}}
\alias{\TR}{[0,T]\times \Rd}
\alias{\djN}{\delta^{j,(N)}}

\author{Kim Weston}
\address[Kim Weston]{Department of Mathematics,
Rutgers University, Piscataway, NJ, USA}
\email{kw552@rutgers.edu}
\author{Gordan {\v Z}itkovi{\' c}}
\address[Gordan {\v Z}itkovi{\' c}]{Department of Mathematics, The
University of Texas at Austin, Austin, TX, USA}
\email{gordanz@math.utexas.edu}
\subjclass[2010]{Primary: 91B51. Secondary: 60H30. JEL Classification: D52,
G12}
\keywords{Incomplete markets, Radner equilibrium, Annuity, BSDE, Systems of
BSDE}

\title[Stochastic annuity in equilibrium]{An incomplete equilibrium
with a stochastic annuity}\thanks{The authors are
grateful to Kasper Larsen for helpful discussions.
  The first author acknowledges the support
  by the National Science Foundation under Grant No.~DMS-1606253.
  The second author acknowledges the support
  by the National Science Foundation under Grant
  No.~DSM-1815017 (2018-2021).
  Any
  opinions, findings and
  conclusions or recommendations expressed in this material are those of the
  authors and do not necessarily reflect the views of the National Science
  Foundation (NSF)}
\date{\today}
\begin{document}

\begin{abstract}
We prove the global existence of an incomplete, continuous-time
finite-agent Radner equilibrium in which exponential agents optimize their
expected utility over both running consumption and terminal wealth.  The
market consists of a traded annuity, and, along with unspanned income,
the market is incomplete. Set in a Brownian framework, the income is
driven by a multidimensional diffusion, and, in particular, includes
mean-reverting dynamics.

The equilibrium is characterized by a system of fully coupled quadratic
backward stochastic differential equations, a solution to which is
proved to exist under Markovian assumptions.
\end{abstract}

 \maketitle

\section{Introduction}

We prove the existence of a Radner equilibrium in an incomplete,
continuous-time finite-agent market setting. The economic agents act as
price takers in a fully competitive setting and maximize exponential
utility from running consumption and terminal wealth.  An annuity in
one-net supply is traded on a financial market, and it pays a constant
running and terminal dividend to its shareholders.  The agents choose
between consuming their income and dividend streams or investing in the
annuity.

\bigskip

Although our setting and the income dynamics are quite general, our
financial market looks relatively simple at first glance. The only available
asset is the annuity, and the agents' only choice at any given moment is
how much to consume, keeping in mind that the only way to transfer
wealth from one time to the next is through the annuity. This apparent
simplicity is quite misleading, since the scarcity of the available
traded assets leads to market incompleteness, a notorious difficulty in
equilibrium analysis. Indeed, the fewer assets the agents have at their
disposal, the less efficient the market becomes and the harder it
becomes to use the standard tools such as the representative agent
approach. In our case, this lack of assets is pushed to its limit.

\bigskip

Admittedly, it would be more realistic to consider markets with several
assets, both risky and riskless, where the incompleteness is derived
from the constraints on each asset's ability to incorporate all the risk
present in the environment. We believe that the exploration of such
problems is one of the most interesting and important areas of future
research in this area. Unfortunately, the formidable mathematical
difficulties present in virtually all such problems leave them outside
the scope of the techniques available to us today.

\bigskip

One of the advantages of our model is its ability to incorporate various
income stream dynamics, including unspanned mean-reverting income
streams (which have been studied extensively for their empirical
relevance; see, e.g.,  \cite{W04JME, W06JME, C14JF}). To the best of
our knowledge, our model is the first with exponential agents to
incorporate unspanned mean-reverting income in equilibrium and prove the
existence of such a equilibrium. The
general income streams we study lead to stochastic
annuity dynamics, which prevent a money market account from being
replicated by trading in the annuity in equilibrium.

\bigskip

Our approach crucially relies on the presence of a traded annuity.  We
also need utility functions of exponential type and a Markovian
assumption on the dynamics of the income streams in order to obtain
conveniently structured individual agent problems, amenable to a BSDE
analysis. Even so, the analysis involves a non-standard Ansatz for the
value function, as we need to formally treat the asset price $A$ as a
quantity that, in standard models, plays the role of a money market
account.  We are not the first to introduce a traded annuity into an
equilibrium model (see, e.g., \cite{VV99ET, C01JET, CLM12JET, CL14RAPS,
W18MAFE}).
Our contribution is to recognize the role of a traded annuity price in
the individual agent value functions, even when general income streams
render the annuity dynamics computationally intractable.

\bigskip

The backward stochastic differential equation (BSDE)/PDE-system approach
to incomplete market equilibria dates back to \cite{Z12FS, Z12phd,
CL15FS, KXZ15wp, XinZit18}, with the early work relying on a smallness-
type assumption on some ingredient of the model (the time-horizon, size
of the endowment, etc.) The mathematical analysis of the present paper
is quite involved  and relies heavily on some recent results of
\cite{XinZit18}, which overcome smallness conditions and treat the
existence and stability of solutions to quadratic systems of BSDE.
Moreover, the applicability of those results in our setting is not at
all immediate and is contingent on a number of a-priori estimates
specific to our model.

\bigskip

\emph{Notation and conventions.} For $J,d\in \N$, The set of $J\times
d$-matrices is denoted by $\R^{J\times d}$. The Euclidean space $\R^J$
is identified with the set of $\R^{J\times 1}$, i.e., vectors in $\R^J$
are columns by default. The $i$-th row of a matrix $Z\in \R^{J\times d}$
will be denoted by $Z^i$, and $\abs{\cdot}$ will denote the Euclidean
norm on either $\R^ {J\times d}$ or $\R^J$.

\medskip

We work on a finite time horizon $[0,T]$ with $T>0$, where
$\FFF=\prf{\sF_t}$ is the usual augmentation of the filtration generated
by a $d$-dimensional Brownian motion $B$.

\medskip

The stochastic integral with respect to $B$ is taken for $\R^{1\times
d}$-valued (row) processes as if $dB$ were a column of its components,
i.e., $\int \sigma(t)\, dB_t$ stands for $\sum_{j=1}^d \int
\sigma_j(t)\, dB^j$. Similarly, for a process $Z$ with values in
$\R^{J\times d}$, $\int Z_t\, dB_t$ is an $\R^{J}$-valued process  whose
components are the stochastic integrals of the rows $Z^i$ of $Z$ with
respect to $dB_t$.

\medskip

For a function defined on a domain in $\R^d$, the derivative $D u$ is
always assumed to take row-vector values, i.e., $Du (x) \in \R^{d\times
1}$. If $u$ is $\R^{J}$ valued, the Jacobian $Du$ will, as usual, be
interpreted as an element of $\R^{J\times d}$. The Hessian, $D^2 u$ of a
scalar-valued function takes values in $\R^{d\times d}$, and we will
have no need for Hessians of vector-valued maps in this paper.

\medskip

To relieve the notation, we omit the time-index from many expressions
involving stochastic processes but keep (and abuse) the notation $dt$
for an integral with respect to the Lebesgue measure.

\medskip

The set of all adapted, continuous and uniformly bounded processes is
denoted by $\sS^{\infty}$, and the set of all processes of bounded mean
oscillation by BMO (we refer the reader to \cite{Kaz94} for all the
necessary background on the BMO processes). The family of all
$B$-integrable processes $\sigma$ such that $\int \sigma\, dB$ is in
$BMO$ is denoted by $\bmo$.

\smallskip

The set of all $\FFF$-progressively measurable process is denoted by
$\sP$. $\sP^r$ denotes the set of all $c\in \sP$ with $\int_0^T
\abs{c}^r\, dt<\infty$, a.s. The same notation is used for scalar,
vector or matrix-valued processes - the distinction will always be clear
from the context.

\medskip

\section{The problem}
\subsection{Model primitives}
The model primitives can be divided into three groups. In the first one,
we
describe the uncertain environment underlying the entire economy. In the
second, we
postulate the form of the dynamics of the traded asset, and in the third we
describe the characteristics of individual agents.  A single real consumption
good is taken as the numeraire throughout.
\subsubsection{The state process}
\label{sss:state}
For $d\in\N$, we start with an $\Rd$-valued \define{state process}
$\xi$
whose dynamics is given by
\begin{align}
\label{equ:xi}
  d\xi_t =
  \Ld(t,\xi_t)\, dt + \Sigma(t,\xi_t)\, dB_t,\
  \xi_0 = x_0\in \Rd
\end{align}
where the measurable functions $\Ld:\TR\to\Rd$ and $\Sigma:\TR\to \R^
{d\times d}$
satisfy the following the regularity assumption:
\begin{assumption}[Regularity of the state process]\label{asm:xi}\
There exists a constant $K>0$ such that for all
$t,t'\in[0,T], x, x'\in \Rd$ and $z\in \Rdi$ we have
\begin{enumerate}[itemsep=1ex,topsep=1ex]
  \item $\abs{\Ld(\tx)}\leq K$ and $\abs{\Ld(t, x)-\Ld(t, x')} \leq
  K
  \abs{x-x'}$,
  \item $\abs{\Sigma(\tx)} \leq K$, and
  $\abs{\Sigma(t, x) - \Sigma(t, x')} \leq K(\sqrt{ \abs{t'-t}}+\abs{x-x'})$
  and
  \item $ \abs{\Sigma(t,x) z} \geq \tfrac{1}{K}\abs{z}$.
 \end{enumerate}
 \end{assumption}
 \begin{remark}
Under Assumption \ref{asm:xi}
the SDE \eqref{equ:xi} admits a unique strong solution.
The
full significance of the assumptions above, however, will only
be apparent in the later sections and is related to the ability to use certain
existence results for systems of backward stochastic differential equations.
\end{remark}

\subsubsection{The traded asset}
Our market consists of a single real asset $A$ in one-net
supply, whose dynamics we
postulate to be of the following form:
   \begin{align}
   \label{equ:dA}
      dA_t =  (A_t \mu_t-1)\, dt  + A_t \sigma_t
      \, dB_t,\ A_T=1,
   \end{align}
with the processes $\mu$ and $\sigma$ to be determined in the equilibrium.
It can be interpreted as
an annuity
which pays a dividend at the rate $1$ during $[0,T]$, as well as a unit
lump sum payment at time $T$.

\bigskip

Let $\Gamma$, the \define{coefficient space}, denote the set of
all
pairs $\gamma=(\mu,\sigma)$, where $\mu$ is a scalar-valued process and
$\sigma$ is an $\Rid$-valued
bmo process.
For simplicity, we often identify the market $A^{\gamma}$
with its coefficient pair $\gamma=(\mu,\sigma)$, and talk simply about
\define{the
market $\gamma$}. The set of all markets given by \eqref{equ:dA}
is not bijectively parametrized by $\Gamma$ as not
every $\gamma\in\Gamma$ defines a
market. Indeed,  the terminal condition $A_T=1$ imposes a nontrivial
relationship between $\mu$ and $\sigma$; for example, if $\mu$ is
deterministic, $\sigma$ either has to vanish or one of its components has to be
truly
stochastic. The set of those $\gm\in\Gamma$ that do define a market is
denoted
by $\Gamma_{f}$ and its
elements are said to be \define{feasible}. If we need to stress that it
comes with feasible coefficients $\gamma\in \Gamma_f$, we write $A^{\gamma}$
for the process given by \eqref{equ:dA} .
\subsubsection{Agents} There are a finite number $I \in \N$ of economic
agents, each of which is characterized by the following four elements:
\begin{enumerate}
  \item the \define{risk-aversion coefficient}
  $\alpha^i>0$. It fully
  characterizes the agent's utility
function $U^i$ which is  of exponential form
\[ U^i(c) = -\oo{\alpha^i} e^{ -\alpha^i c} \efor c\in\R.\]
\item the \define{random-endowment (stochastic income) rate}.
  Each agent receives an endowment of the consumption good at the rate
  $e^i_t=e^i(t,\xi_t)$ and a lump sum $e^i_T=e^i(T,\xi_T)$
  at time $T$, for some function $e^i:[0,T]\times \Rd \to \R$.
  \item the \define{initial holding} $\pi^i_0\in \R$ is the initial number
  of shares of the annuity $A$ held by the agent.
\end{enumerate}

With the \define{cumulative endowment rate} defined by
\[ e = \sumi e^i,\]
we impose the following regularity conditions:
\begin{assumption}[Regularity of the endowment rates]\
\label{asm:e}
\label{asm:main}
\begin{enumerate}
\item
Each $e^i$ is bounded and continuous, while its terminal section
$e^i (T,\cdot)$ is $\alpha$-H\"older continuous for some $\alpha \in (0,1]$.
\item
The cumulative endowment process
$e_t=e(t,\xi_t)$ is a
semimartingale with the decomposition
\[ e(t,\xi_t) = e(0,x_0)+\int_0^t \mue(s,\xi_s)\, ds + \int_0^t
\sigmae(s,\xi_s)\, dB_s,\]
where the drift function $\mue:[0,T]\times \Rd \to \R$ is bounded
and continuous and
$\sigmae(s,\xi_s)$ is a bmo process.
\end{enumerate}
\end{assumption}
We will
often abuse notation and write $e^i$  both for the function
$e^i:[0,T]\times \Rd \to \R$ and the stochastic process $e^i_t=e^i(t,
\xi_t)$.
The same
applies to other functions applied to $(t,\xi_t)$ - such as $e$ or
$\mue$.
\begin{remark}\label{rem:scope}
It is worth stopping here to give a few examples of state processes
$\xi$ and the functions $e^i$ which satisfy all the regularity
conditions imposed so far. Once the coefficients $\Ld$ and $\Sigma$ for $\xi$ are
picked so as to satisfy Assumption \ref{asm:xi}, then Assumption \ref{asm:e}
is easy to check for a sufficiently smooth $e^i$ by a simple application of
\itos{} formula.

\medskip

The more interesting observation is that there is room for improvement. It may
seem that the boundedness imposed in Assumption \ref{asm:xi}
rules out some of the most important classes of state
processes such as the classical mean-reverting (Ornstein-Uhlenbeck)
processes. This is not the case, as we have the freedom to choose both
the state process $\xi$, and the deterministic function $e^i$ applied to
it, while only caring about the resulting composition. We illustrate
what we mean by that with a simple example. The reader will easily add
the required bells and whistles to it, and adapt it to other
similar frameworks.

\smallskip

We assume that $d=1$ and that
we are interested in the random endowment rate $e^i(t,\eta_t)$ where
$e^i$ is a bounded and appropriately smooth function, and $\eta_t$ is
an Ornstein-Uhlenbeck process with the dynamics
\begin{align*}
  d\eta_t = \theta (\bar{\eta} - \eta_t)\, dt + \sigma_{\eta}\, dB_t,
\end{align*}
and parameters $\theta,\sigma_{\eta}>0$ and $\eta_0,\bar{\eta} \in \R$.
Since the drift function $x\mapsto \theta(\bar{\eta} - x)$ is not bounded, the
process $\eta$ does not satisfy the conditions of Assumption \ref{asm:xi}.
The process $\eta$ admits, however, an explicit expression in terms of a
stochastic integral of a deterministic process with respect to the underlying
Brownian motion:
\begin{align}
\label{equ:explicit}
   \eta_t = \bar{\eta} + (\eta_0 - \bar{\eta}) e^{-\theta t}+
    \sigma_{\eta} e^{-\theta t} \int_0^t e^{\theta s}\, dB_s
\end{align}
If we define the state process $\xi$  by
\begin{align*}
 d\xi_t = e^{-\theta t}\, dB_t, \ \xi_0=0,
\end{align*}
i.e., if we set $\Ld(t,x) = 0$ and $\Sigma(t,x) = e^{-\theta t}$, the
boundedness
of the time horizon $[0,T]$ allows use to conclude that $\Ld$ and $\Sigma$
satisfy Assumption \ref{asm:xi}.
Moreover, by
\eqref{equ:explicit}, the choice
$f^i(t,x) = e^i(t,\bar{\eta} + (\eta_0 - \bar{\eta}) e^{-\theta t} +
\sigma_{\eta} x)$ yields
\begin{align*}
   f^i(t,\xi_t) = e^i(t,\eta_t).
\end{align*}
This way, we can represent a function of an interesting, but not entirely
compliant state process $\eta$ as a (modified) function of a regular state
process $\xi$. The boundedness (and other regularity properties) of the
function $e^i$ are inherited by $f^i$, thanks to the boundedness from above and
away from zero of the function $t\mapsto e^{-\theta t}$.
\end{remark}

\subsection{Admissibility and equilibrium}
\begin{definition}
\label{def:64F9}
Given a feasible set of coefficients $\gamma=(\mu,\sigma)\in\Gamma_f$,
a pair $(\pi,c)$ of scalar processes is said to be  a \define{$\gm$-admissible
strategy} for agent $i$ if
\begin{enumerate}
  \item $\abs{c}+\abs{\pi (A^{\gm}\mu-1)} \in \sP^1$ and $\pi A^{\gm} \sigma
  \in \bmo$.
\item the \define{gains process} $X=X^{\pi,\gm}=
\pi A^{\gm}$ is a semimartingale which satisfies the \define{self-financing condition}
\[ dX =  \pi\, dA^{\gm} +(e^i - c + \pi  )\,
dt .\]
\end{enumerate}
The set of all $\gm$-admissible strategies for agent $i$
is denoted by $\sA_{\gm}^i$, and
the subset of $\sA_{\gm}^i$ consisting of the strategies with $\pi(0)=\pi_0$
is given by $\sA_{\gm}^i(\pi_0)$.
\end{definition}

\begin{definition}
\label{def:equilibrium}
We say that $\gamma^*\in\Gamma_f$ is the set of \define{equilibrium market
coefficients} (and $A^{\gm^*}$ an \define{equilibrium market}) if
there exist $\gamma^*$-admissible strategies $( \hpi^i, \hc^i) \in
\sA^i_{\gm^*}(\pi^i_0)$,
$i=\ft{1}{I}$,  such that the following two conditions hold:
\begin{enumerate}
  \item \emph{Single-agent optimality:}
  For each $i$ and all
  $(\pi,c) \in \sA^i_{\gm^*}(\pi_0^i)$ we have
  \[\textstyle
  \EE[ \int_0^T U^i(\hc^i_t)\, dt] + \EE[ U^i(X^{\hpi^i,\hc^i}_T + e^i_T)]
  \geq
  \EE[ \int_0^T U^i(c_t)\, dt] + \EE[ U^i(X^{\pi,c}_T + e^i_T)]\]
  \item \emph{Market clearing:}
  \begin{align*}
    \sumi \hpi^i = 1 \eand \sumi \hc^i = e + 1 \text{ on } [0,T), \eand
    \sumi X^{\hpi^i,\hc^i}_T = 1, \text{ a.s.}
  \end{align*}
  \end{enumerate}
\end{definition}
\section{Results}
\subsection{A BSDE characterization}
Our first result is a
characterization of equilibria in terms of a system
of
backward stochastic differential equations (BSDE). These systems
consist of $1+I$ equations, with the first component generally playing a
different role from the other $I$. For that reason, it pays to
depart slightly from the classical
notation $(Y,Z)$, where $Y$ has as many components as there are equations, and
the driver $Z$ is a matrix process whose additional dimension reflects the
number of driving Brownian motions. Instead,  we
use the notation $\big((a,Y), (\sigma,Z)\big)$ where $a$ is a scalar and $Y$ is
$\R^{I\times 1}$-valued. Similarly, $\sigma$ and $Z$ are $\Rid$- and $\R^
{I\times d}$-valued processes, respectively. As usual, we say that $((a,Y),
(\sigma,Z))$ is an $(\sS^{\infty} \times \bmo)$-solution if all the components
of $a$ and $Y$ are in $\sS^{ \infty}$, and all components of $(\sigma,Z)$ are
in $\bmo$.
To simplify the notation, we also
introduce the following, derived, quantities:
\begin{equation}\label{eqn:constants} \oo{\ba} := \sumi \oo{\alpha^i},
\eand   \kappa^i :=
\tf{\ba}{\alpha_i}>0 \text{ so that } \sumi \kappa^i=1.\end{equation}
\begin{theorem}[A BSDE Characterization]
\label{thm:char}
Suppose that $\sumi \pi^i_0=1$, that
Assumption \ref{asm:main} holds, and that
 $\big((a,Y),(\sigma,Z)\big)$ is
 an $\sS^{\infty}\times \bmo$-solution to
\begin{equation}
 \begin{aligned}
 \label{equ:BSDE}
    da   &= \sigma\, dB + \Big( \ba \mue - \tot \suml \kappa^l
    \abs{Z^l}^2 - \exp(-a) \Big)\, dt, &a_T
    & = 0 \\
    dY^i &= Z^i\, dB  + \Big( \tot \abs{Z^i}^2 + \exp(-a)( 1 + a + Y^i -
    \alpha^i e^i ) \Big)\, dt, &Y^i_T &= \alpha^i e^i_T,\ 1\leq i \leq I.
 \end{aligned}
 \end{equation}
Then $A=\exp(a)$ is an equilibrium annuity price with
market coefficients $(\mu,\sigma) \in \Gamma_f$, where $\mu$ is given by
 \begin{align}
 \label{equ:mu}
   \mu = \ba \mue +  \tot \abs{\sigma}^2
     - \frac{1}{2}\sumi \kappa^i \abs{Z^i}^2.
 \end{align}
\end{theorem}
\begin{remark}We note that the validity of Theorem \ref{thm:char} above does
not depend on
Assumption \ref{asm:xi}. In fact, no Markovian assumption is needed for it, at
all.
Moreover, the full force of Assumption \ref{asm:e} is not needed, either.
It would be enough to assume that each $e^i$ is in bmo and that the cumulative
endowment process $e$
is a semimartingale of the form $de = \mue\, dt + \sigmae\, dB$, where
$\mue$ and $\sigmae$ are  general bmo processes and not necessarily
bounded functions of
a state process.
\end{remark}
\begin{proof}
 Having fixed an $(\sS^{\infty},\bmo)$-solution $\big( (a,Y), (\sigma,
 Z)\big)$, we
 set $A=\exp(a)$ and define $\mu$ as in
 \eqref{equ:mu}, so that $A$ satisfies \eqref{equ:dA}.
With the market coefficients $\gamma=(\mu,\sigma)$ fixed,
we pick an agent  $i\in \set{1,\dots, I}$ and a pair $(\pi,c) \in \sA^i_{\gm}
(\pi^i_0)$,
and define processes $X^i$,  $\tV^i$ 
and  $V^i$ by
\[ X^i = \pi A,\   \tV^i = -
\exp(-\alpha^i X^i/A - Y^i) \eand
V^i=\tV^i+
 \int_0^{\cdot} -\exp(-\alpha^i c_t) \,dt.  \]
The self-financing property of  $(\pi,c)$ implies that the
semimartingale decomposition of $V^i$ is given by
$dV^i = \mu_V\, dt + \sigma_V\, dB$,
where
 \begin{align*}
 \mu_V = -\exp(-\alpha^i c)  +\tf{-\tV^i}{A} \big(
 1-\log( \tf{-\tV^i}{A})  \big)
- \big( \alpha^i c   \big)
\tf{-\tV^i}{A}\ \eand\ \sigma_V = -\tV^i Z^i.
 \end{align*}
Young's inequality implies that $\mu_V \leq 0$ and that
the coefficients
$\mu_V$ and $\sigma_V$ are regular enough
to conclude that $V^i$ is a supermartingale for all
admissible $(\pi,c)$. Therefore,
 \begin{equation*}
 \begin{split}
 \EE[ \int_0^T U^i(c_s)\, ds] &+ \EE[ U^i(X^i_T+e^i_T)] =\\
 &=
\oo{\alpha^i}
\EE[ \int_0^T -\exp(-\alpha^i c_s ) \, ds]  + \oo{\alpha^i} \EE[ \exp(-
\alpha^i
(X^i_T  + e^i_T))] \\
&= \oo{\alpha^i} \Big( \EE[\int_0^T -\exp(-\alpha^i c_s )\, ds ]+ \EE[
\exp(-\alpha^i X^i_T/A_T -Y^i_T  )] \Big)\\
&=
\oo{\alpha^i}\EE[ V^i_T] \leq \oo{\alpha^i} V^i_0 = -\oo{\alpha^i}
\exp(-\alpha^i  \pi^i_0   - Y^i_0).
\end{split}
\end{equation*}
Next, in order to characterize the optimizer, we construct a
consumption process for
which $\mu_V=0$. More precisely, we let the process $\hX^i$  be the
unique solution of
the
following linear SDE:
 \begin{align}
 \label{equ:hX-SDE}
   \hX^i_0= \pi^i_0 A_0,\ d\hX^i = \Big(\mu \hX^i + (e^i  - \oo{\alpha^i}
   (a+Y^i)- \tf{\hX^i}{A})\Big)\, dt+\hX^i \sigma\ dB,
 \end{align}
and set
\[ \hc^i =
\oo{\alpha^i} (a+Y^i)+ \tf{\hX^i}{A} ,\ \hpi^i =
\tf{\hX^i}{A}.\]
It follows immediately that $(\hpi^i, \hc^i) \in \sA_{\gm}^i(\pi^i_0)$ and that
the process $\hX$ is the associated gains process. The choice of $\hc^i$,
through $\hX^i$, makes the process $V^i$ a martingale and the pair $(\hpi^i,
\hc^i)$ optimal for agent $i$.

\medskip

Turning to market clearing,
we consider the process $F= a+\sumi \kappa^i Y^i - \ba e $, whose dynamics
are given by
\begin{align}
\label{equ:F-dyn}
dF &= (\sigma + \sumi \kappa^i Z^i - \ba \sigmae)\,dB +
 \exp(-a) F \, dt,\ F_T=0.
\end{align}
In other words, the pair $(Y,\zeta)=(F,\sigma + \sumi \kappa^i Z^i - \ba
\sigmae)$ is an
$\sS^{\infty}\times\bmo$-solution to the linear BSDE
\begin{align*}
dY = \zeta\, dB + \exp(-a)Y\, dt, \ Y_T=0.
\end{align*}
Since $a$ is bounded,
the coefficients of this BSDE are globally Lipschitz, and, therefore,
by the uniqueness
theorem (see \cite[Theorem 4.3.1, p.~84]{Zha17}), we can
conclude
that
$F=0$. That implies that
\[ a+\sumi \kappa^i Y^i = \ba e \text{ on } [0,T], \]
and, so,
\[ \sumi \hc^i = e + \oo{A} \sumi \hX^i.\]
The form of the dynamics \eqref{equ:hX-SDE} of each $\hX^i$ leads to the
following dynamics for $\hX = \sumi \hX^i$:
 \begin{align}
 \label{equ:SDE2}
   d\hX = (\mu \hX - \oo{A} \hX)\, dt + \hX \sigma\, dt.
 \end{align}
The assumption that $\sum \pi^i_0 =1$ implies that $\hX_0 = A_0$, which, in
turn, implies that the process $A$ is also a solution to \eqref{equ:SDE2}. By
uniqueness, we must have $\hX = A$ and conclude that the clearing
conditions are satisfied.
\end{proof}
\subsection{Existence of an equilibrium.}
Next, we show that under additional assumptions on the problem
ingredients - most notably that of a Markovian structure - the characterization
of Theorem \ref{thm:char} can be used to establish the existence of an
equilibrium market.
\begin{theorem}
Under \label{thm:exist}
Assumptions \ref{asm:xi} and \ref{asm:main}, the system \eqref{equ:BSDE}
admits an $\sS^ {\infty} \times \bmo$-solution.
\end{theorem}
The BSDE characterization of Theorem \ref{thm:char} immediately implies the
main result of the paper:
\begin{corollary}
Under Assumptions \ref{asm:xi} and \ref{asm:main}, there exists a set
$\gamma^*=(\mu^*,\sigma^*)$ of feasible market coefficients such that $A^
{\gamma^*}$ is an equilibrium market.
\end{corollary}

\begin{proof}[Proof of Theorem \ref{thm:exist}]
In certain situations it will be convenient to standardize the notation, so we
also
write $Y^{0}$ for $a$, $Z^{0}$
for $\sigma$, and set
  \[ g^i(x) =
  \begin{cases}
    0, & i=0, \\
    \alpha^i e^i(T,x), & 1\leq i \leq I.\\
  \end{cases}\]
   The $dt$-terms in \eqref{equ:BSDE} define the driver
$f:[0,T]\times \Rd \times \R^{I+1} \times \R^{(I+1) \times d} \to
\R^{I+1}$ in the usual way:
\begin{align*}
    f^0(t,x,y,z) & =\ba \mue(t,x) -
    \tot \suml \kappa^l \abs{z^l}^2 - \exp(-y^0),\\
    f^i(t,x,y,z) & = \tot \abs{z^i}^2 + \exp(-y^0)\Big( 1 + y^0 + y^i -
    \alpha^i e^i(t,x) \Big), \efor i=\ft{1}{I}.
\end{align*}
The system \eqref{equ:BSDE}, written in the new notation, becomes
\begin{align}
\label{equ:BSDE-proof}
   dY^i_t = f^i(t,\xi_t,Y_t,Z_t)\, dt + Z^i_t\, dB_t\, \ Y^i_T = g^i(\xi_T),\
   i= \ft{0}{I}.
\end{align}

\emph{Step 1 (truncation).}
We start by truncating the driver $f$ to obtain a sequence of
well-behaved, Lipschitz problems.
More precisely,
given $N>0$ we define
\[ \iota_N(x) = \max(\min(x,N),-N) \efor x\in\R \eand q_N(z) = \abs{z} \iota_N(
\abs{z}), \efor z\in \R^{1\times d},\]
so that $\iota_N$ and $q_N$ are
Lipschitz functions with Lipschitz constants $1$ and $N$, respectively.
Moreover,
\[ \abs{\iota_N(x)}\leq N \eand \abs{q_N(z)} \leq N \abs{z}.\]

\medskip

Using the functions defined above, for each $N\in\N$ we pose a truncated
version of \eqref{equ:BSDE}:
\begin{equation}
 \label{equ:BSDE-N}
 \tag{BSDE${ }_N$}
   \begin{aligned}
   da &= \sigma\, dB +
   \Big(
   \ba \mue  - \tot \suml \kappa^l q_N(Z^l) -
   \exp(-\ioN{a}) \Big)\, dt,\\
d Y^i  &= Z^i  \, dB +
 \Big(
  \tot q_N(Z^i)
 + \exp(-\ioN{a})( 1 + \ioN{a} + \ioN{Y^i}  - \alpha^i e^i)
 \Big)\, dt
   \end{aligned}
\end{equation}
with the terminal conditions $Y^i_T=\ioN{ g^i_T}$ and $a_T=0$. We define the
driver $(t,x,y,z) \mapsto \upN{f}(t,x,y,z)$ from the $dt$-terms in the standard
way.

\medskip

For each $N\in\N$, $\upN{f}$ is continuous in all of its variables,
  uniformly Lipschitz in both $z$ and $y$,  and $\upN{f}(t,x,0,0)$ is
  bounded.
  Assumption \ref{asm:xi} guarantees that the same is true for the function
  $\upN{F}(t,x,y,z) = - \upN {f}(t,x,y,z \Sigma^{-1}(t,x))$.
  Therefore, we can apply
  Proposition \ref{pro:F-Lip} in the Appendix to conclude that
  there exists a solution $(\YN, \ZN)$ to \eqref{equ:BSDE-N} of the form
  \begin{align*}
  \YN_t = \vN(t,\xi_t),\ \ZN_t=\wN(t,\xi_t),
  \end{align*}
  with $\vN:\TR\to\R^{I+1}$
  bounded and $\wN:\TR\to \R^{(I+1)\times d}$ such
  that $\ZN$ is a bmo process.
  We note that existence for
  \eqref{equ:BSDE-N}  is also guaranteed by the
 classical result \cite[Theorem 3.1, p.~58]{ParPen90}, but only in the
 class $\sS^2\times \sH^2$, which is too big for our purposes.

\bigskip
\emph{Step 2 (uniform estimates).}  The bounds guaranteed by Proposition
\ref{pro:F-Lip} all depend on the truncation constant $N$, so our
next task is to explore the special structure of our system and
establish bounds in terms of universal quantities.  A universal constant,
in this proof, will be a quantity that depends on the constants $\alpha^i$,
 the time-horizon $T$ and the $\sS^{\infty}$-bounds on $e^i$ and $\mue$,
but not on $N$. We denote such a constant
by $C$, and allow it to change from line to line.

\medskip

Let $\big( (\aN, \YiN), (\sigN, \ZN)\big)$ be the solution to the truncated
system from Step 1~above.
It follows from the dynamics of
$\aN$  and the fact that $q_N(z)\geq 0$ for all $z\in\Rid$
 that $\aN -\int_0^{\cdot} \ba \mue\, dt$ is a supermartingale, so
that for all $t\in[0,T]$,
\[ \aN_t  \geq \EE[ \aN_T -\int_t^T \ba \mue\, dt | \sF_t]
\geq  -(T-t)\norm{\ba \mue}_{\sS^{\infty}}, \text{ i.e., } \aN_t \geq
-C.\]

 Next, we turn to $\YN$ and use the fact that the components of
 $Y$ are coupled only through $a$.
 This way, we can get uniform bounds on $\YiN$ if
 we manage to produce a uniform bound on the function of $a$ appearing on
 the right-hand side. We start by using the
 following easy-to-check inequality
 \[ \exp(-x)(1+\abs{x})\leq \exp(2x^-),\text{ for all } x\in\R,\]
and the fact that $(\ioN{x})^- \leq (x)^-$ for all $x$, to obtain that for all $t\in[0,T]$,
  \[ \exp(-\ioN{\aN_t})\Big(1+\babs{\ioN{\aN_t}}\Big)\leq C.\]
It is readily checked that there exist a bounded measurable
function
$\upN{\delta}: \Rid \to \Rd$, such that
\begin{align*}
     q_N(z) = z \upN{\delta}(z) = \sum_{j=1}^d z_j \djN(z), \efor z=(z_1,\dots,
     z_j)\in \Rid.
\end{align*}
 Therefore, for each $i=\ft{1}{I}$, there exists a probability measure $\PP^i\,
 (=\PP^{i,N}) \sim\PP$ under which the process
 $\tilde{B}^i = B + \int_0^\cdot \upN{\delta}(\ZiN)\, dt$ is a
 Brownian motion on
 $[0,T]$. Since $\ZiN$ is guaranteed to be in $\bmo$, it remains in $\bmo$
 under the measure $\PP^i$ (see \cite[Theorem 3.3, p.~57]{Kaz94}). Therefore,
  the process
 $\int \ZiN\, dB^i$ is a
 $\PP^i$-martingale and we can take the
 expectation of the $i$-th
 equation with
 respect to $\PP^i$ to obtain
 \begin{align*}
   \babs{\YiN_t}
   &\leq \babs{\EE^i[\ioN{\alpha^ie^i(T)}| \sF_t]}
     +\int_t^T\EE^i\left[\exp(-\ioN
   {\aN_s})\Big(1+\babs{\ioN{\aN_s}}\Big)| \sF_t\right]\, ds\\
   &\ \ \ \ \ +\int_t^T\EE^i\left[\exp(-\ioN{\aN_s})\babs{\ioN
   {\YiN_s}-\alpha^ie^i_s}\, | \sF_t\right]\, ds\\
   &\leq C\Big(1+\int_t^T\EE^i[\babs{\YiN_s}\ | \sF_t]\, ds\Big)
   \leq C\Big(1+\int_t^T y^i(s)ds\Big),
 \end{align*}
 where
 $y^i(t)= \norm{\YiN_t}_{\linf}$. Thus, $y^i$ satisfies
 \[ y^i(t) \leq C(1+\int_t^T y^i(s)\, ds), \eforall t\in [0,T],\]
 for some universal constant $C$.
   Gronwall's inequality implies that $y^i(0)=\norm{\YiN}_{\Sin}$ is bounded
   by another universal constant, so we conclude that there exists a
   universal $\Sin$-bound on all $\YiN$.

\medskip

Our next goal is to produce universal bmo bounds on the processes $\ZiN$. This
will follow by using the universal boundedness of the $Z$-free
terms in the driver of $\YiN$ obtained above. Since the $i$-th component of the driver
$\upN{f}$ depends on $\ZN$ only through $\ZiN$, for $1\leq i \leq I$,
we can apply standard
   exponential-transform estimates.  We adapt
\label{page:exp-tranform}
   the argument in \cite[Proposition 2.1, p. 2925]{BriEli13} and define
   $$
     \phi(x) := \frac{\exp(2|x|)-1-2|x|}{4} \efor x\in\R,
   $$
   noting that both $\phi$ and $\phi'$ are nonnegative and increasing, while
   $\phi\in C^2(\R)$ with
   $\phi''-2|\phi'|=1$.
   Thus, for any stopping time $\tau$ in $[0,T]$, It\^o's Lemma gives us that
   \begin{align*}
     0&\leq \phi(\YiN_\tau)
     \leq \EE[\phi(\YiN_T)|\sF_\tau] +\EE\left[\int_\tau^T \babs{\phi'(\YiN_s)}C\left(1+\|\YiN\|_
     {\Sin}\right) |\sF_\tau\right]   \\
     & \qquad + \EE\left[
     \int_\tau^T\left(\babs{\frac{1}{2}\phi'(\YiN_s)}
     \babs{\ZiN_s}^2-\phi''(\YiN_s)\babs{\ZiN_s}^2\right)ds\ |\sF_\tau\right]\\
     &\leq \phi(\|\YiN\|_{\Sin}) + \\
     & \qquad +C\int_0^T\phi'(\|\YiN\|_{\Sin})(1+\|\YiN\|_
     {\Sin})ds
       - \EE\left[\int_\tau^T \babs{\ZiN_s}^2ds\ |\sF_\tau\right].
   \end{align*}
   Rearranging terms yields
   \begin{align*}
     \EE&\big[\int_\tau^T\babs{\ZiN_s}^2ds\Big| \mathcal{F}_\tau\Big]
     \leq \phi(\|\YiN\|_{\Sin})+\\
     & \qquad +C\int_0^T\phi'(\|\YiN\|_{\Sin})(1+\|\YiN\|_
     {\Sin})ds.
   \end{align*}
The right-hand side admits a universal bound (independent of $N$ and $\tau$),
and, hence,
so does the left-hand side.

\medskip

Finally, we go back to the equation satisfied by $\aN$ and note that the
term $\exp(-\ioN{\aN})$ is bounded because $(\aN)^-$ is.
We can bound $\aN$ from above in an $N$-independent manner,
by a combination of the $\bmo$-bounds on $\ZN$ and the sup norm of $\mue$.
By taking expectations and using universal boundedness/$\bmo$-property of
all the other terms, we conclude that $\upN{\sigma}$ also admits a universal
$\bmo$-bound.

\medskip

Having the universal bounds on $\YN$ and $\aN$, we can remove some of the
truncations introduced in \eqref{equ:BSDE-N}. Indeed, for $N$ larger than
the largest of the $\sS^{\infty}$-bounds on $\YN$ and $\aN$, we have
\[ \ioN{\YiN} = \YiN \eand \ioN{\aN} = \aN.\]
Therefore, there exists a constant $N_0$ such that for $N\geq N_0$
the processes $(\YN,\aN)$ together with $(\ZN, \upN{\sigma})$ solve the intermediate system
\begin{equation}
 \label{equ:BSDE-iN}
 \tag{BSDE${}_N'$}
   \begin{aligned}
   da &= \sigma\, dB +
   \Big(
   \ba \mue  - \tot \sum_j \kappa^j q_N(Z^j)
   -
   \exp(-\ioNz{a}) \Big)\, dt,\\
   d Y^i  &= Z^i  \, dB +
 \Big(
   \tot q_N(Z^i)
 + \exp(-\ioNz{a})( \ioNz{Y^i} + \ioNz{a}  - \alpha^i e^i+1)
 \Big)\, dt \\
   \end{aligned}
\end{equation}
with the same terminal conditions as \eqref{equ:BSDE}.

\bigskip

\emph{Step 3 (Bensoussan-Frehse conditions and the existence of a Lyapunov
function).}
Mere boundedness in $\sS^{\infty} \times \bmo$ is not sufficient to
guarantee subsequential convergence of the solution $(\YN,\aN)$ of the
truncated system to a limit which solves
\eqref{equ:BSDE} or \eqref{equ:BSDE-iN}. It has been shown, however, in
\cite[Theorem 2.8, p.~501]{XinZit18},
that an
additional property - namely the existence of a uniform Lyapunov function -
will guarantee such a convergence. The existence of such a
function can be deduced from another result of the same paper,
\cite[Proposition 2.11, p.~503]{XinZit18}, once its conditions are checked.
This proposition states that a uniformly bounded sequence of solutions of a
sequence of BSDE such as \eqref{equ:BSDE-iN} admits a common Lyapunov function
if the structure of its drivers satisfies the so called
Bensoussan-Frehse conditions uniformly in $N$
(see \cite[Definition 2.10, p~502]{XinZit18} for the definition).
It applies here because our system is of upper-triangular form
when it comes to its quadratic dependence on $z$. More precisely, the
driver of the system \eqref{equ:BSDE-iN} can be represented as a sum
of two functions $\upN{f}_1$ and $\upN{f}_2$ given by
 \begin{align*}
    (\upN{f}_1)^i(t,x,y,z) &= \begin{cases}
   \ba \mue(t,x) - \exp(-\ioNz{y^{0}}) & i=0 \\
   \exp(-\ioNz{y^{0}})( \ioNz{y^i} +
   \ioNz{y^{0}}  - \alpha^i e^i(t,x)+1) \Big), &  1\leq i \leq I \\
   \end{cases}\\
(\upN{f}_2)^i(t,x,y,z) &= \begin{cases}
   -\tot \sum_{l=1}^I \kappa_l q_N(z^l) & i=0\\
   \tot q_N (z^i)   & 1\leq i \leq I
   \end{cases}
 \end{align*}
where the convention that $a=Y^{0}$ and $\sigma = Z^{0}$ is used.
 Therefore, there exists a universal constant $C$ such that, for all $1\leq
 i \leq I+1$, we have
\[ \abs{(\upN{f}_1)^i (t,x,y,z)} \leq C\]
as well as
\[\abs{(\upN{f}_2)^i(t,x,y,z)} \leq
C(1+ \sum_{j=1}^i \abs{q_N(z^j)})
\leq
C(1+ \sum_{j=1}^i \abs{z^j}^2)
.\]
Therefore, $\upN{f}$ can be split into a subquadratic (in fact bounded) and
an upper triangular component, allowing us to conclude that a uniform
Lyapunov function for $(\upN{f})_{N\geq N_0}$ can be constructed.

\bigskip

\emph{Step 4 (Passage to a limit).} It remains to use
\cite[Theorem 2.8, p.~501]{XinZit18} to conclude that a subsequence of
$\upN{v}$ converges towards a
continuous function $v:\TR\to\R^{I+1}$ such that $Y_t
= v(t,\xi_t)$ and $Z_t=Dv(t,\xi_t)$ solves the limiting system
\begin{equation}
 \label{equ:BSDE-lim}
 \tag{BSDE'}
   \begin{aligned}
   da &= \sigma\, dB +
   \Big(
   \ba \mue  - \tot \sum_l \kappa^l \abs{Z^l}^2
   -
   \exp(-\ioNz{a}) \Big)\, dt. \\
d Y^i  &= Z^i  \, dB +
 \Big(
   \tot \abs{Z^i}^2
 + \exp(-\ioNz{a})( \ioNz{Y^i} + \ioNz{a}  - \alpha^i e^i+1)
 \Big)\, dt,
   \end{aligned}
\end{equation}
for $i=\ft{1}{I}$,
where, as above $a=Y^0$ and $\sigma=Z^0$.
As far as the conditions of Theorem 2.8 in \cite{XinZit18} are concerned, the
most difficult one,
 the existence of a Lyapunov function,
has been settled in Step 3.~above. The other conditions - the uniform H\"
older boundedness of the terminal conditions, and a-priori boundedness - are
easily seen to be implied by our standing assumptions. Finally, since $Y$
is a pointwise limit of a sequence of functions bounded by $N_0$, the same
processes $(Y,Z)$ also solve the original BSDE \eqref{equ:BSDE} (without
 truncation at $N_0$).
\end{proof}

\section{Bounded solutions of Lipschitz quasilinear systems} The main result of
this section, Proposition \ref{pro:F-Lip},
collects some results on systems of
heat equations with Lipschitz nonlinearities on derivatives up to the
first order. We suspect that these results may be well-known to PDE
specialists, but we were unable to find a precise reference under the
same set of assumptions in the literature, and, therefore, decided to
include a fairly self-contained proof.

\medskip

In the sequel, $D$ denotes the derivative operator with respect to all
spatial variables, i.e., all variables except $t$.
For $d,J\in \N$ and $\beta\geq 0$, we define the following three Banach spaces:
\begin{enumerate}
  \item $\linf=\linf(\Rd, \R^J)$ or $\linf=\linf(\Rd; \R^{J\times d})$,
  depending on the context,
  \item
$W^{1,\infty} = W^{1,\infty}(\Rd; \R^{J})$, with the norm $\norm{U}_{\Wi}
= \norm{U}_{\linf} + \norm{DU}_{\linf}$.
\item $\lob = \lob( [0,T); \Wi)$ - the Banach space of
measurable
functions $u:[0,T]\to \Wi$, endowed with the exponentially weighted norm
\[ \norm{u}_{\lob} =  \int_0^T e^{-\beta (T-t)} \norm{u(t,\cdot)}_{\Wi}\, dt.\]
\end{enumerate}
 The infinitesimal generator of the state process $\xi$ is given by
 \begin{align*}
      \sA u(t,x) = Du(t,x) \Lambda(t,x) + \tot \Tr\Big(D^2 u(t,x) \, \Sigma
      (t,x)
      \Sigma^T(t,x) \Big) \efor (t,x) \in \TR.
 \end{align*}

\begin{proposition}
\label{pro:F-Lip}
Suppose that
$g:\Rd \to \R^J$ and $F:[0,T]\times \Rid \times \R^{J} \times \R^
{J\times
d} \to \R^J$ are
measurable functions such that
\begin{itemize}
  \item $\abs{g(x)} \leq M$,
  \item $\abs{F(t,x,0,0)}\leq M$, and
  \item $\abs{F(t,x,y_2,z_2) - F(t,x,y_1,z_1)} \leq M( \abs{y_2-y_1} +
  \abs{z_2-z_1})$,
\end{itemize}
for some $M$ and all $t,x,y_1,y_2, z_1, z_2$, and that
the functions $\Ld$ and $\Sigma$
satisfy the conditions of Assumption \ref{asm:xi} (with the constant $K$).
Then  the following statements hold:
\begin{enumerate}
\item The PDE system
 \begin{align}
 \label{equ:PDE-lip}
   u_t + \sA u + F(\cdot,\cdot,u,Du)=0, \ u(T,\cdot) = g
 \end{align}
admits a weak solution $u$ on $[0,T]$. Moreover $u(t,\cdot) \in
W^{1,\infty}$ for all $t\in [0,T)$ and there exists a constant
$C=C(J,d,M,T,K) \in [0,\infty)$  such that
\[ \norm{u(t,\cdot)}_{\linf} \leq C
\eforall t\in [0,T]
\eand \int_0^T \norm{D
u(t,\cdot)}_{\linf}\, dt\leq C\]
\item Let $u$ denote a solution of \eqref{equ:PDE-lip} as in (1) above and
let
$\prf{\xi_t}$ be a strong solution of the SDE
\begin{align*}
     d\xi_t = \Ld(t,\xi_t)\, dt + \Sigma(t,\xi_t)\, dB_t.
\end{align*}
The pair $(Y_t,Z_t)$, where $Y_t = u(t,\xi_t)$ and $Z_t
= Du (t,\xi_t) \Sigma(t,\xi_t)$ is an $\sS^{\infty} \times \bmo$-solution
to the system
\begin{align}
\label{equ:two}
  dY^i_t = - F^i\Big(t,\xi_t, Y_t, Z_t \Sigma^{-1}(t,\xi_t)\Big)\, dt + Z^i_t\,
  dB_t,\
  Y^i_T = g^i(\xi_T),\ i=\ft{1}{I}
\end{align}
\end{enumerate}
\end{proposition}

\begin{proof}
Throughout the proof, $C$ will denote a constant which may depend on $J,d,M,T$
or
$K$,
but not on $\beta$, $t,s$ or $x$,
and can change from line to line; we will call such a constant universal. The
assumptions on $F$
imply that uniformly in $t$,  and for all $U,V\in \Wi$, we have
 \begin{align}
 \label{equ:Lip-F}
    \norm{F(t,\cdot, U, DU) -
   F(t,\cdot, V,DV)}_{\linf} &\leq C \norm{U-V}_{W^{1,\infty}}, \eand\\
 \label{equ:bd-F}
    \norm{F(t,\cdot,U,DU)}_{\linf} &\leq C (1+\norm{U}_{\Wi}).
 \end{align}
Let $p(\tx;\sxp)$ denote a fundamental solution associated to the
operator $u_t + \sA u$, i.e., $(\tx)\mapsto p(\tx,\sxp)$ solves
\begin{align*}
     p_t + \sA p &= 0 \text{ for } \tx \in [0,s) \times \Rd
\end{align*}
classically
and satisfies the boundary condition $\lim_{t\nearrow s} \int_{\Rd} \psi(x) p
(t,x,s,x')\,dx = \psi(x')$ for each bounded and continuous $\psi$.
We refer the reader to
\cite[Theorem 10, p.~23]{Fri64} and the discussion
preceding it for existence of a positive fundamental solution under the
conditions of Assumption \ref{asm:xi}. Moreover, the
equations (6.12) and (6.13) on p.~24 of \cite{Fri64} state that
there exist universal constants $C,\ld>0$ such that  for $t<s$ and all $x,x'$
we have
\begin{align}
\label{equ:G-bounds}
\abs{p(\tx,\sxp)} \leq C \vp_{\ld}(\tx,\sxp) \eand
\abs{\partial_{x_k} p(\tx,\sxp)} \leq C \oo{\sqrt{s-t}}
\vp_{\ld} (\tx,\sxp),
\end{align}
for all $k=\ft{1}{d}$, where
\begin{align*}
     \vp_{\ld}(\tx;\sxp) = \oo{(2\pi \ld^2 (s-t))^{d/2}} \exp\Big( - \oo{2
     \ld^2 (s-t)} \abs{x'-x}^2\Big),
\end{align*} is the
scaled heat kernel (which is, itself, a
fundamental solution associated to $u_t+\tot \ld^2 \Delta u$.).

These properties, in particular, allow us to define
the function $\Phi[u]:[0,T]\times \Rd \to \R^{J}$ by
 \begin{align}
 \label{equ:def-Phi}
    \Phi[u](t,x) = \int_{\Rd} \int_t^T F\Big(s,x',
   u(s,x'), Du (s,x')\Big) p(t,x;s,x')\, ds\, dx',
 \end{align}
 for each $u \in \lob$.
The equation \eqref{equ:bd-F} guarantees that  $\Phi[u]$ is well-defined with
$\Phi[u](t,\cdot) \in \linf$.

The Gaussian bounds in \eqref{equ:G-bounds} imply that one can pass the
derivative under the
integral sign to obtain
 \begin{align}
 \label{equ:D-Phi}
    \partial_{x_k}\Phi[u](t,x) = \int_{\Rd} \int_t^T F\Big(s,x',
   u(s,x'), Du (s,x')\Big) \partial_{x_k} p(\tx;\sxp)
   \, ds\, dx',
 \end{align}

Consequently $t\mapsto\Phi[u](t,\cdot)$ is an a.e-defined measurable map $
[0,T]\to
\Wi$, for each
$u\in \lob$.
To bound the norm of $\Phi[u]$ we start with the following estimate,
fueled by \eqref{equ:G-bounds},
 \begin{align*}
   \norm{ \Phi[u](t,\cdot)}_{\Wi}
   &\leq C
    \int_t^T (1+\norm{u(s,\cdot)}_{\Wi})
     \int_{\Rd} \left(p(\tx;\sxp) + \sum_{k=1}^J \abs{\partial_{x_k} p
     (\tx;\sxp)}\right)\, dx' \, ds \\
     &  \leq C \int_t^T \oo{\sqrt{s-t}} (1+\norm{u(s,\cdot)}_{\Wi})\, ds.
 \end{align*}
Furthermore, \eqref{equ:def-Phi} and \eqref{equ:D-Phi} imply
that
 and, so,
 \begin{align*}
 \norm{\Phi[u]}_{\lob} &\leq C \int_0^T
    e^{\beta(t-T)} \int_t^T \oo{\sqrt{s-t}}
    (1+\norm{u(s,\cdot)}_{\Wi})\, ds\,
    dt\\
    &= C\int_0^T (1+\norm{u(s,\cdot)}_{\Wi})\, ds\,
    \int_0^s \oo{\sqrt{s-t}} e^{\beta (t-T)}\, dt\\
    &\leq \tf{C}{\sqrt{\beta}}(1 + \norm{u}_{\lob})
 \end{align*}
A similar computation also yields
 \begin{align}
 \label{equ:Phi-Lip}
    \norm{ \Phi[u]- \Phi[v]}_{\lob} &\leq \tf{C}{\sqrt{\beta}} \norm{u-v}_{\lob}.
 \end{align}

Next, for $g\in \linf$, we define
\[ \Psi[g](t,x) = \int_{\Rd} g(x') p(t,x; T,x')\, dx',\]
so that, as above,
\[ \norm{\Psi[g](t,\cdot)}_{\Wi} \leq \tf{C}{\sqrt{T-t}}
\norm{g}_{\linf} \text{ and }
 \norm{\Psi[g]}_{\lob} \leq \tf{C}{\sqrt{\beta}} \norm{g}_{\linf},\]
 and $\Psi[g]\in\lob$ for each $g\in \linf$. Therefore, the function
 \[ \Gamma[u] = \Phi[u] + \Psi[g],\]
 maps $\lob$ into $\lob$ and
 \eqref{equ:Phi-Lip} implies that it is Lipschitz,
 with constant $C/\sqrt{\beta}$. Since $C$ does not depend on $\beta$,
 we can turn $\Gamma$ into a contraction by
choosing a large-enough $\beta$, and conclude that
$\Gamma$ admits a unique fixed point $u \in \lob$. The integral representations in
\eqref{equ:def-Phi} and \eqref{equ:D-Phi} allow us to conclude that
  $u$ and $Du$ are continuous functions on $[0,T)\times \Rd$.
  Moreover, thanks to the Markov property of $\xi$, we have
  \[ u(t,\xi_t)  = \EE[ g(\xi_T)+
  \int_t^T f(s,\xi_s)\, ds | \sF_t], \text{
  a.s.}\]
where
 \begin{align}
 \label{equ:f-from-F}
   f(s,x) = F(s,x,u(s,x), Du(s,x)) \in \R^J.
 \end{align}
Since $\norm{f(t,\cdot)}_{\linf}\leq C(1+\norm{u(t,\cdot)}_{\Wi})$ for all $t$,
the map $t\mapsto
\norm{u(t,\cdot)}_{\Wi}$ belongs to $\lob$, and  (stripped of its norm)
the space $\lob$ does not
depend on the choice of $\beta$. Therefore,
\[ \Bnorm{u(t,\xi_t) - \EE[ g(\xi_T)|\sF_t] }_{\linf} \leq \int_t^T
\norm{f(s,\cdot)}_{\linf}\, ds \to 0 \text{ as } t\to T.\]
Since $g$ is bounded, we have $\norm{u(t,\cdot)}_{\linf}\leq
C$ for all $t$. Moreover, the martingale $\EE[ g(\xi_T)|\sF_t]$ admits a
continuous modification, so the process $Y$, defined by
\[ Y_t =
\begin{cases}
  u(t,\xi_t), & t<T\\
  g(\xi_T), & t=T,
\end{cases} \]
is a.s.-continuous.
This allows us to conclude, furthermore, that $Y_t + \int_0^t f(s,\xi_s)$
is a continuous modification of the martingale
\[ M_t = \EE[ g(\xi_T) + \int_0^T f(s,\xi_s)\, ds|\sF_t],\]
making $Y$ a semimartingale. To show that $(Y,Z)$ as in the statement indeed
solves \eqref{equ:two}, we need to argue that
the martingale $M_t-M_0$ must be of the
form $\int_0^t Du(s,\xi_s) \Sigma(s,\xi_s)\, dB_s$. This can be proven by
approximation as in the proof of \cite[Lemma 4.4, p.~516]{XinZit18}).

\bigskip

The last step is to argue that $(Y,Z)$ is an $\sS^{\infty}\times
\bmo$-solution. The function $u$ is uniformly bounded, so it suffices to
establish the $\bmo$-property of $Z$. This can be bootstrapped from the
boundedness of $Y$ by applying \itos{} formula to the bounded processes
$\exp(c Y^i)$, $i=\ft{1}{J}$, for large-enough constant $c$. A similar argument
is already presented on page \pageref{page:exp-tranform}, in the proof of
Theorem \ref{thm:exist}, so we skip the details.
\end{proof}

\def\cprime{$'$}
\providecommand{\bysame}{\leavevmode\hbox to3em{\hrulefill}\thinspace}
\providecommand{\MR}{\relax\ifhmode\unskip\space\fi MR }
\providecommand{\MRhref}[2]{%
  \href{http://www.ams.org/mathscinet-getitem?mr=#1}{#2}
}
\providecommand{\href}[2]{#2}


\begin{thebibliography}{CLM12}

\bibitem[Cal01]{C01JET}
Laurent~E. Calvet, \emph{Incomplete markets and volatility}, Journal of
  Economic Theory \textbf{98} (2001), no.~2, 295--338.

\bibitem[CL14]{CL14RAPS}
Peter~O. Christensen and Kasper Larsen, \emph{Incomplete continuous-time
  securities markets with stochastic income volatility}, Review of Asset
  Pricing Studies \textbf{4} (2014), no.~2, 247--285.

\bibitem[CL15]{CL15FS}
Jin~Hyuk Choi and Kasper Larsen, \emph{Taylor approximation of incomplete
  {R}adner equilibrium models}, Finance and Stochastics \textbf{19} (2015),
  no.~3, 653--679.

\bibitem[CLM12]{CLM12JET}
Peter~Ove Christensen, Kasper Larsen, and Claus Munk, \emph{Equilibrium in
  securities markets with heterogeneous investors and unspanned income risk},
  Journal of Economic Theory \textbf{147} (2012), no.~3, 1035--1063.

\bibitem[Coc14]{C14JF}
John~H. Cochrane, \emph{A mean-variance benchmark for intertemporal portfolio
  theory}, Journal of Finance \textbf{69} (2014), no.~1, 1--49.

\bibitem[EB13]{BriEli13}
Romuald Elie and Philippe Briand, \emph{A simple constructive approach to
  quadratic {BSDEs} with or without delay}, Stochastic Processes and their
  Applications \textbf{123} (2013), no.~8, 2921--2939.

\bibitem[Fri64]{Fri64}
Avner Friedman, \emph{Partial differential equations of parabolic type},
  Prentice-Hall, Inc., Englewood Cliffs, N.J., 1964.

\bibitem[Kaz94]{Kaz94}
Norihiko Kazamaki, \emph{Continuous exponential martingales and {BMO}}, Lecture
  Notes in Mathematics, vol. 1579, Springer-Verlag, Berlin, 1994.

\bibitem[KX{\v{Z}}15]{KXZ15wp}
Constantinos Kardaras, Hao Xing, and Gordan {\v{Z}}itkovi{\'c},
  \emph{Incomplete stochastic equilibria with exponential utilities close to
  {P}areto optimality}, Submitted for publication, 2015.

\bibitem[PP90]{ParPen90}
{\'E}.~Pardoux and S.~G. Peng, \emph{Adapted solution of a backward stochastic
  differential equation}, Systems Control Lett. \textbf{14} (1990), no.~1,
  55--61.

\bibitem[VV99]{VV99ET}
Dimitri Vayanos and Jean-Luc Vila, \emph{Equilibrium interest rate and
  liquidity premium with transaction costs}, Economic Theory \textbf{13}
  (1999), 509--539.

\bibitem[Wan04]{W04JME}
Neng Wang, \emph{Precautionary saving and partially observed income}, Journal
  of Monetary Economics \textbf{51} (2004), no.~8, 1645--1681.

\bibitem[Wan06]{W06JME}
\bysame, \emph{Generalizing the permanent-income hypothesis: Revisiting
  {F}riedman's conjecture on consumption}, Journal of Monetary Economics
  \textbf{53} (2006), no.~4, 737--752.

\bibitem[Wes18]{W18MAFE}
Kim Weston, \emph{Existence of a {R}adner equilibrium in a model with
  transaction costs}, Mathematics and Financial Economics \textbf{12} (2018),
  no.~4, 517--539.

\bibitem[X{\v{Z}}18]{XinZit18}
Hao Xing and Gordan {\v{Z}}itkovi{\'c}, \emph{A class of globally solvable
  {M}arkovian quadratic {BSDE} systems and applications}, Ann.~Probab
  \textbf{46} (2018), no.~1, 491--550.

\bibitem[Zha12]{Z12phd}
Yingwu Zhao, \emph{Stochastic equilibria in a general class of incomplete
  {B}rownian market environments}, Ph.D. thesis, University of Texas at Austin,
  2012.

\bibitem[Zha17]{Zha17}
Jianfeng Zhang, \emph{Backward stochastic differential equations}, Probability
  Theory and Stochastic Modeling, vol.~84, Springer, 2017.

\bibitem[{\v{Z}}it12]{Z12FS}
Gordan {\v{Z}}itkovi{\'{c}}, \emph{An example of a stochastic equilibrium with
  incomplete markets}, Finance and Stochastics \textbf{16} (2012), no.~2,
  177--206.

\end{thebibliography}
\end{document}